\documentclass[lettersize,journal]{IEEEtran}
\usepackage[caption=false,font=normalsize,labelfont=sf,textfont=sf]{subfig}
\usepackage{cite}
\usepackage[final]{graphicx}
\usepackage{amsmath}
\usepackage{amssymb}
\usepackage[linesnumbered,ruled,vlined]{algorithm2e}
\usepackage{amsfonts}
\usepackage{array}
\usepackage{booktabs}
\usepackage{stfloats}
\usepackage{lineno}
\usepackage{bm}
\usepackage{stfloats}
\usepackage{epstopdf}
\usepackage{amsthm}
\usepackage{color}
\usepackage{multirow}
\usepackage{multicol}
\usepackage{lipsum}
\newcolumntype{C}{>{\centering\arraybackslash$}p{\linewidth}<{$}}

\usepackage[linesnumbered,ruled,vlined]{algorithm2e}
\theoremstyle{plain}
\newtheorem{theorem}{Theorem}

\newtheorem{remark}{Remark}

\begin{document}
	\title{Conceal Truth while Show Fake: T/F Frequency Multiplexing based Anti-Intercepting Transmission}
	\author{Zhisheng~Yin,~\IEEEmembership{Member~IEEE,}
		Nan~Cheng,~\IEEEmembership{Senior Member~IEEE}, 
%		Xiaojie~Fang,~\IEEEmembership{Member Member~IEEE}, 
		Mingjie~Wang,
		Changle~Li,~\IEEEmembership{Senior Member~IEEE},
		Wei~Xiang,~\IEEEmembership{Senior Member~IEEE}
		\thanks{
This work was supported in part by the National Natural Science Foundation of China under Grant 62201432, Grant 62071356, and Grant 62101429.

Z. Yin, N. Cheng and C. Li are with State Key Lab. of ISN and School of Telecommunications Engineering, Xidian University, Xi'an 710071, China (e-mail: zsyin@xidian.edu.cn; dr.nan.cheng@ieee.org; clli@mail.xidian.edu.cn).

M. Wang is with Academy for Network \& Communications of China Electronics Technology Group Corporation (CETC), Shijiazhuang 050081, China (e-mail: dtwmingjie@163.com).

W. Xiang is with the School of Engineering and Mathematical Sciences, La
Trobe University, Victoria 3086, Australia (e-mail: W.Xiang@latrobe.edu.au)

.

\textit{Corresponding authors}: Nan Cheng.

%	This work was supported in part by the Fundamental Research Funds for the Central Universities of Ministry of Education of China under Grant XJS221501, the National Natural Science Foundation of Shaanxi Province under Grant 2022JQ-602, National Natural Science Foundation of China (No. 62071356 and 6210010668), Natural Science Foundation of Jiangsu Province  (No. BK20200440), and Natural Sciences and Engineering Research Council (NSERC) of Canada.			
	}}%
	\maketitle
	
	\IEEEpeerreviewmaketitle
\begin{abstract}
In wireless communication adversarial scenarios, signals are easily intercepted by non-cooperative parties, exposing the transmission of confidential information. 
This paper proposes a true-and-false (T/F) frequency multiplexing based anti-intercepting transmission scheme capable of concealing truth while showing fake (CTSF), integrating both offensive and defensive strategies. Specifically, through multi-source cooperation, true and false signals are transmitted over multiple frequency bands using non-orthogonal frequency division multiplexing. The decoy signals are used to deceive non-cooperative eavesdropper, while the true signals are hidden to counter interception threats. Definitions for the interception and deception probabilities are provided, and the mechanism of CTSF is discussed.
To improve the secrecy performance of true signals while ensuring decoy signals achieve their deceptive purpose, we model the problem as maximizing the sum secrecy rate of true signals, with constraint on the decoy effect. Furthermore, we propose a bi-stage alternating dual-domain optimization approach for joint optimization of both power allocation and correlation coefficients among multiple sources, and a Newton's method is proposed for fitting the T/F frequency multiplexing factor. In addition, simulation results verify the efficiency of anti-intercepting performance of our proposed CTSF scheme.
\end{abstract}
\section{Introduction}
In the realm of wireless communications, the security of transmitted information is of paramount importance, particularly in scenarios where communication signals are susceptible to interception by non-cooperative entities \cite{Ahuja2021}. In military applications, secure communication is essential to prevent adversaries from intercepting strategic information \cite{Abughalwa2020}. In industrial contexts, it is critical for safeguarding sensitive data against espionage. For civilian applications, it ensures the protection of personal data and privacy \cite{Guri2023}. As technologies such as the Internet of Things (IoT) and smart cities continue to advance, robust security measures are indispensable to mitigate potential threats and vulnerabilities in these non-cooperative environments \cite{Du2023}. Addressing passive threats such as interception and eavesdropping, and conducting research on secure transmission in electronic reconnaissance and electromagnetic adversarial scenarios, represents a significant challenge in both military and industrial protection domains\cite{Wang2021,Roy2022}. Moreover, this is an urgent area for breakthroughs in academic research.

Unlike upper-layer encryption algorithms and security protocols, which rely on key distribution and computational complexity, electromagnetic space security focuses on electromagnetic countermeasures to ensure secure transmission over wireless channels \cite{Venkatesh2021,Sahin2022}. It primarily addresses threats such as electronic interference, signal eavesdropping, deception, electromagnetic shielding, and spectrum theft. To resist interference, interception, and deception, electromagnetic security employs dense overlapping of wave propagation across spatial, temporal, frequency, and energy domains for comprehensive control, protection, and management \cite{Yan2023,Nguyen2021}.
Techniques like frequency hopping, spread spectrum, and beamforming complicate interception and eavesdropping \cite{Takeshita2024,Yin2023}. 
Additionally, physical layer security measures, including jamming-resistant protocols, artificial noise, and secure waveform design, enhance the robustness of wireless communications against passive threats \cite{Zhang2021,Yin2022}.

 \begin{figure}[t]
	\centering
	\includegraphics[width=0.45\textwidth]{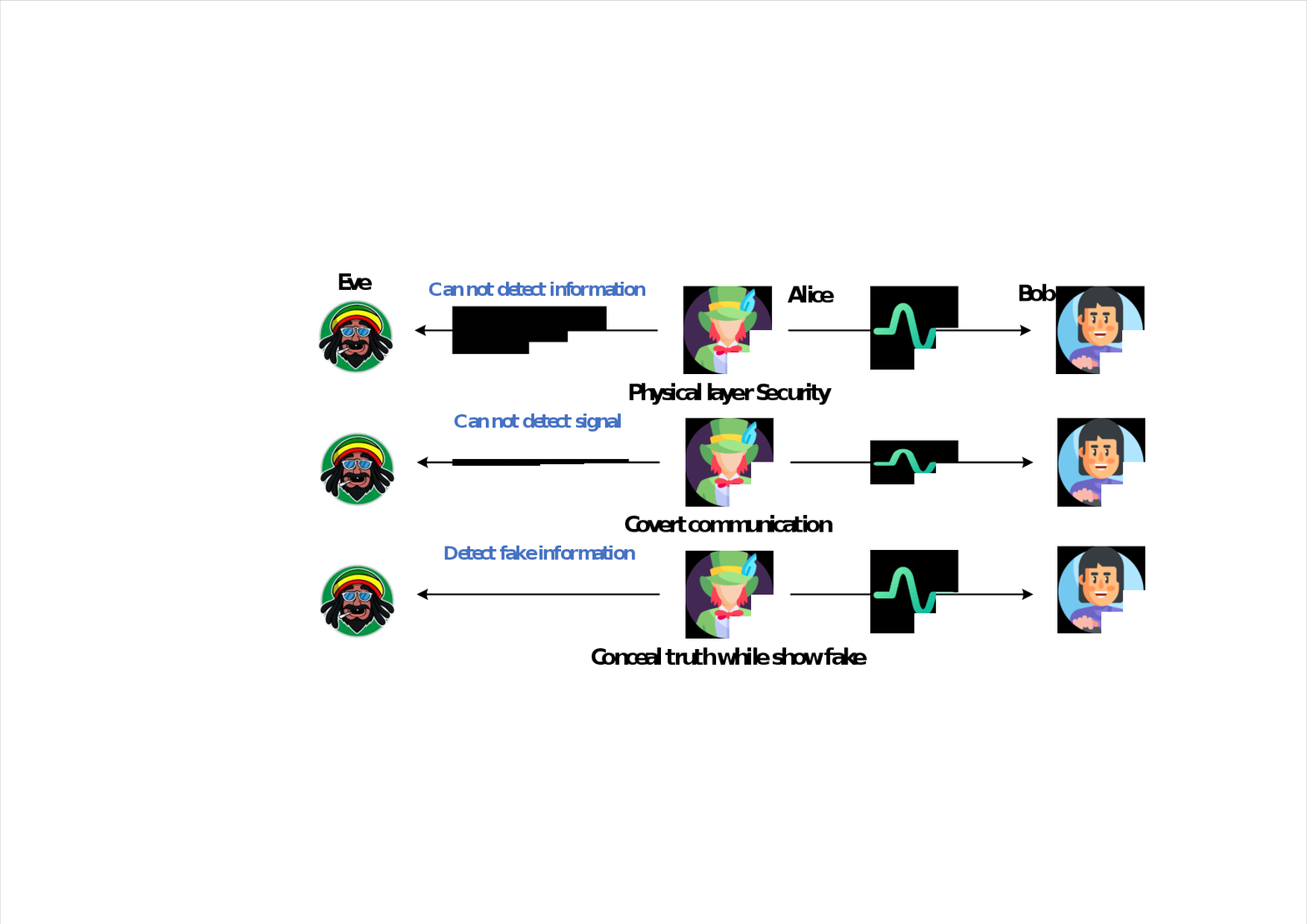}
	\caption {Anti-intercepting/secure transmission schemes}
	\label{fig01}
\end{figure}

Physical layer security addresses passive eavesdropping threats by utilizing the random differences between wireless channels through signal processing methods to achieve secure transmission. Its theoretical performance metrics often include secrecy capacity, secrecy rate, and secrecy outage probability \cite{Yin2023,Yin2022,Liu2023}, etc. However, most of the current work primarily focuses on how to enhance secrecy capacity, using techniques such as resource optimization, artificial noise, and relay cooperation, along with extensive theoretical analysis. In contrast, there is still a lack of substantial work on how to achieve the attainable secrecy capacity through secure coding or other methods, which remains significant challenges.
Additionally, covert communication aims to further conceal the communication signal from being detected by adversaries, such as by using low power spectral density transmission schemes to hide the signal within the ambient background noise \cite{Li2024,Jiang2021,Zhang2021a}. 
However, due to its power-constrained nature, the coverage capability, transmission rate, and reliability of the signal are susceptible to degradation, especially in the presence of interference. Therefore, achieving anti-interception/secure transmission in wireless communication adversarial scenarios still presents significant challenges.
 In addition, considering a strongly adversarial communication environment, such as military unmanned reconnaissance communication systems, it is essential to ensure that one's own communication links possess both anti-interception capabilities for secure communication and the ability to counter adversarial strikes \cite{Li2023}. In the context of IoT secure communication, adversarial attacks and defense requirements in 6G and AI-assisted IoT systems have also attracted significant attention \cite{Son2024}. Therefore, offensive and defensive collaborative anti-interception communication methods are required.

By investigating and analysis on the related work of the aforementioned studies, it is evident that anti-interception transmission schemes are highly targeted. Different approaches have their own advantages and disadvantages in various adversarial scenarios and purposes. Physical layer security methods can ensure a high secure transmission capacity but often overlook stealth, focusing on the security of the information layer. Covert communication schemes emphasize the anti-detection capability of the communication behavior itself, offering high covertness but often having limited transmission rates. Inspired by existing work, we focus on the ultimate goal of achieving victory in adversarial scenarios and aim to develop a novel anti-interception transmission scheme.

The comparison of anti-intercepting transmission schemes is illustrated in Fig. \ref{fig01}. Different from existing secure transmission schemes that are studied solely from a defensive perspective, this paper first proposes an integrated offensive and defensive anti-interception transmission scheme: conceal truth while show fake (CTSF). 
Particularly, our proposed scheme, termed True-and-Fake (T/F) frequency multiplexing, employs non-orthogonal frequency division multiplexing to transmit both true and fake signals across multiple frequency bands, where true signals are transmitted on ``true frequencies" and deceptive fake signals are transmitted on ``fake frequencies".  
We mask the true signals within the fake signals by T/F frequency multiplexing to conceal the true information, while simultaneously using the fake signals to deceive non-cooperative eavesdroppers. 
The core challenge addressed in this paper is to maximize the secure transmission rate of the true signals while ensuring that the decoy signals effectively serve their purpose of deception. 
To realize this, main contributions of our work are as follows:
\begin{itemize}
	\item  The T/F frequency multiplexing-based CTSF model is proposed as an optimization task aimed at maximizing the sum secrecy rate of cooperative sources, subject to constraints on the reception quality of the decoy signals. This ensures that the decoy signals effectively deceive potential interceptors while maintaining the secrecy of true signals.
	\item To address the non-convex and intractable problem, we propose a bi-stage alternating dual-domain optimization approach. This approach jointly optimizes the multi-source power allocation and the T/F frequency correlation coefficients, balancing the trade-off between the secrecy performance and the deception effectiveness. A Newton's method is also provided for fitting the frequency multiplexing parameter using such correlation coefficients. In addition, the complexity of the proposed algorithms are analyzed.
	\item Insightful discussions are conducted through theoretical analysis. To verify the anti-interception performance of our proposed CTSF scheme, extensive simulations are carried out. Both anti-intercepting performance and deceptive performance are evaluated, confirming the efficiency of our proposed CTSF scheme.  
\end{itemize}

\section{Related Work}
With the advancement of physical layer security research in recent years, extensive studies have addressed secure transmission issues across various application scenarios, including networks in ground, sea, air, and space domains, underscoring the critical importance of secure transmission. Different scenarios present unique security implementation challenges. For instance, in satellite communications, the wide coverage characteristic often results in similar channel conditions for both cooperative and non-cooperative users within the same beam. This similarity poses challenges for physical layer security schemes that rely on channel randomness differences. In different electromagnetic or network environments, the resources available for signal processing can vary, leading to differences in secure transmission methods. Specifically, we will examine the following typical scenarios.

To achieve secure transmission in satellite communications, considering the high mobility and long propagation characteristics of low Earth orbit (LEO) satellites, an uncertainty model for the satellite channel angle and a norm-bounded model for the eavesdropper were established, and a robust transmission design was proposed \cite{Jiang2023}. Achieving secure transmission by exploiting the asymmetrical characteristics of channel state information through beamforming and relay\cite{Lyu2023,Yin2023a}. Secure transmission problems are frequently modeled as the optimization of beamforming or relay selection to maximize the secrecy rate.
In wireless communications, various interferences such as inter-beam interference, inter-user interference, and inter-symbol interference often pose challenges to reliable transmission. However, these interferences can be strategically exploited to mitigate eavesdropping threats and enhance secure transmission design \cite{ Albayrak2023,Wang2024}. In addition, there is the artificial noise method, where artificial noise is actively transmitted in the null space of the main channel to degrade the eavesdropper's signal reception quality \cite{Zhou2023,Roth2024}. Particularly, Tampere University conducted an experiment on the design of a full-duplex jamming transceiver \cite{Marin2023}. This transceiver transmits frequency modulated continuous wave (FMCW) signals, commonly seen in low-cost radar systems, to prevent eavesdroppers from correctly interpreting wireless local area network (WLAN) signals while still being able to receive the same signals. This jammer uses the transmitted swept waveform in the down-converter, causing self-interference from antenna coupling and nearby reflections to degrade into a fixed low-frequency tone. To implement secure transmission, artificial noise is used to generate physical layer security keys to counter eavesdropping attacks \cite{Wei2023}. Key distribution is employed for signal modulation, allowing both keys and data to be transmitted over a single channel, thereby reducing costs and ensuring better compatibility.

Unlike traditional physical layer security methods, covert communications focus more on the anti-detection capability of the communication signal. Common methods include power control, interference injection, and relay selection \cite{Jiang2024}. Covert communication is a technique that hides signals within environmental noise, making it difficult for an eavesdropper to detect the presence of communication \cite{Chen2023}.
Using Kullback-Leibler divergence and total variation distance as metrics for covertness, low detection probability communication is achieved by optimizing the transmit power and interference power. The covert rate under massive multiple input multiple output (MIMO) block fading channels is analyzed \cite{Bai2023}. By utilizing spectrum multiplexing to hide device to device (D2D) signals within cellular network signals and the adversary's detection error minimization problem is conducted \cite{Feng2024}. The methods of relay selection and interference injection are also applicable in covert communications \cite{Forouzesh2020,Su2021}.
To maximize the covert transmission rate while accounting for the adversary target's maneuvering altitude, a method is proposed that jointly designs the radar waveform and communication transmit beamforming vector based on two channel state information (CSI) models, effectively turning interference into an ally \cite{Wang2024a}. Lower and upper bounds on the information-theoretically optimal throughput are derived as functions of the channel parameters, the desired level of covertness, and the amount of shared key available \cite{Zhang2021b}.

There are also some studies that consider both secrecy capacity and covertness. To simultaneously counter two types of adversaries, one that captures data and another that detects communication, a joint method of time slot selection and interference injection is proposed to maximize the secrecy capacity under covertness constraints \cite{Forouzesh2023}. Additionally, by employing a beamforming method based on random artificial noise (AN), the covert rate of NOMA users can be maximized while ensuring secure performance constraints \cite{Li2023}. Based on the aforementioned research summary, existing secure transmission methods such as physical layer security and covert communication ensure secure transmission by focusing on different aspects of information security and anti-detection, primarily from the perspective of passive defense, enabling communication systems to resist interception.

The remainder of this paper is structured as follows: Section II provides a comprehensive overview of related work in the field of secure communications. Section III describes the proposed system model for T/F frequency multiplexing-based multi-source cooperative anti-intercepting transmission. Section IV presents the solution and discusses its implications. Section V details the simulations and presents the numerical results. Finally, Section VI concludes the paper.

%By introducing this innovative T/F Frequency Multiplexing method, we aim to contribute a significant advancement in the field of secure wireless communications, offering a robust solution to counteract the ever-evolving threats of signal interception.

\emph{Notations:} $ \left|  \cdot  \right| $ represents the absolute value operation.
  $\mathcal{CN}\left( \mu, \delta^2\right) $ denotes the complex Gaussian distribution with mean $\mu$ and variance $\delta^2$. $ \Pr \left(  \cdot  \right) $ represents the probability calculation. Other notations are summarized in Table I.     

\begin{table}[h]
	\centering
	\caption{Summary of Notations and Definitions}
	\label{tab11}
	\begin{tabular}{lll}
		\toprule
		\midrule
		Notation  & Definition \\
		\midrule
		T/F & True/Fake \\
		${\cal K}$            &  the set for true frequency for transmitting true signals\\
		$\mathcal{\tilde K} $ & the set for fake frequency for transmitting decoy signals\\
		$P_s$ & Total transmission power\\
		$ { x_i},  i \in \mathcal{ K} $ & true signal with confidential information\\
		$ {\tilde x_i}, i \in \mathcal{\tilde K} $ & decoy signal for confusing the Eve\\
		$ {h_i}$ & Channel response from the $i^{th}$ user to Bob \\
	$ {h_{e,n}}$ & Channel response from the $n^{th}$ user to Eve \\
		$ K $ & The number of sources \\ 
		$ {\gamma _k} $ & The SINR received at Bob for detecting true signal \\ 
		$ {\gamma _{e,k}} $ & The SINR received at Eve for intercepting the true signal \\ 
		$ {{{\tilde \gamma }_{e,n}}}$ & The SINR of deception received at Eve\\
		$ {{\tilde T}_h} $ & A certain threshold \\
		$\alpha$& frequency multiplexing factor\\
		$ c_i $ & The correlation coefficient between signals of sources\\
		\bottomrule
	\end{tabular}
\end{table}

\section{System Model}%and Problem Statement 
 \begin{figure}[ht]
	\centering
	\includegraphics[width=0.48\textwidth]{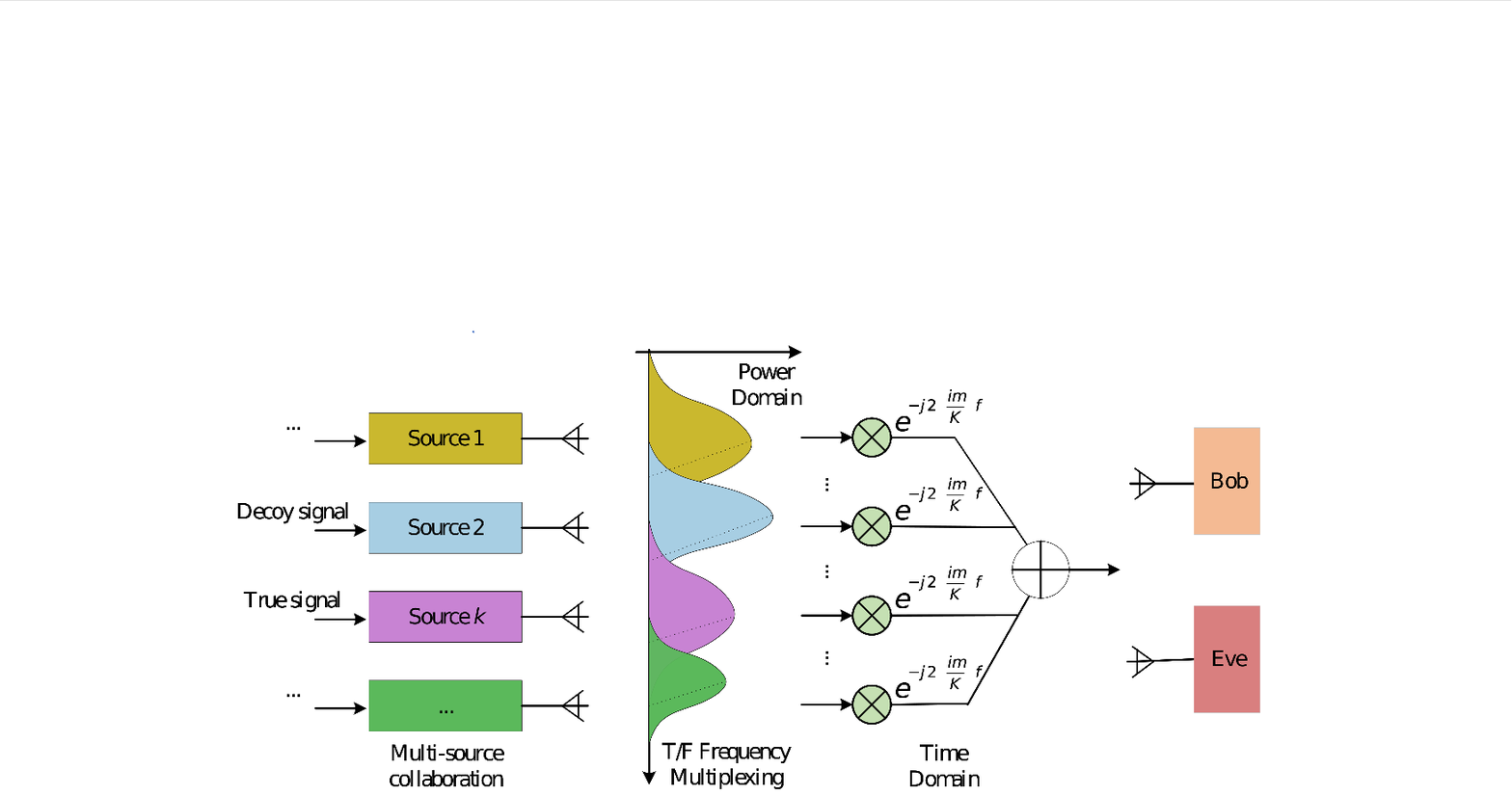}
	\caption {{T/F frequency multiplexing based CTSF anti-intercepting transmission}}
	\label{fig1}
\end{figure}	

In this work, we investigate a multi-source cooperative anti-interception transmission system, where multiple sources jointly and synchronously transmit true and fake information using different frequency bands, as depicted in Fig. \ref{fig1}. 
The proposed system employs a non-orthogonal frequency division scheme for cooperative signal transmission, fundamentally differing from conventional orthogonal frequency division multiplexing (OFDM) systems that maintain strict inter-signal orthogonality. In our multi-source collaborative anti-interception framework utilizing T/F frequency multiplexing, the $ K $ available frequency bands can be strategically allocated such that $ M $ bands carry the true signal while the remaining $ K-M $ bands transmit decoy signals, following either an interleaved distribution pattern or a security-optimized allocation strategy that dynamically balances anti-interception requirements, decoy deception effectiveness, and maintainable transmission rates. The intentional non-orthogonal band allocation creates controlled frequency overlapping and inter-signal interference, resulting in deep time-domain mixing of true and decoy signals that obscures their distinguishability.
%—while legitimate receivers can successfully demodulate the desired signal through known interference cancellation techniques (e.g., successive interference cancellation), eavesdroppers face irreducible demodulation ambiguity due to the unknown band allocation pattern.

In this idea, the true signal is hidden within the fake signals, while the fake signals are used to decoy non-cooperative Eves, thereby achieving the goal of conceal truth while show fake. Without loss of generality, it is assumed that Eve has the same working model as Bob and is capable of eavesdropping on the signals of all sources. In addition, the channel model adopts a Rician fading model in this work.

The signal received by Bob in the $ k^{th} $ frequency band is represented as:
%\begin{equation}\label{eq1}
%	{y_k} = \frac{1}{K}\sum\limits_{m = 0}^{K - 1} {\sum\limits_{i = 0}^{K - 1} {\sqrt {{p_i}} {h_i}{x_i}{e^{j2\pi \frac{{(i - k)m}}{K}\alpha }}} } ,
%\end{equation}

\begin{align}\label{eq1}
		{y_k} &= \frac{1}{K}\sum\limits_{m = 0}^{K - 1} {\sum\limits_{i = 0}^{K - 1} {\sqrt {{p_i}} {h_i}{x_i}{e^{j2\pi \frac{{(i - k)m}}{K}\alpha }}} }  + {n_k}\nonumber\\
		&= \sqrt {{p_k}} {h_k}{x_k} + \frac{1}{K}\sum\limits_{m = 0}^{K - 1} {\sum\limits_{i = 0,i \ne k}^{K - 1} {\sqrt {{p_i}} {h_i}{x_i}{e^{j2\pi \frac{{(i - k)m}}{K}\alpha }}} }  + {n_k}\nonumber\\
		&= \sqrt {{p_k}} {h_k}{x_k} + \underbrace {\frac{1}{K}\sum\limits_{m = 0}^{K - 1} {\sum\limits_{i \ne k}^{\cal K} {\sqrt {{p_i}} {h_i}{x_i}{e^{j2\pi \frac{{(i - k)m}}{K}\alpha }}} } }_{true} \nonumber\\
		& \quad + \underbrace {\frac{1}{K}\sum\limits_{m = 0}^{K - 1} {\sum\limits_{i \ne k}^{\tilde {\cal K}} {\sqrt {{p_i}} {h_i}{{\tilde x}_i}{e^{j2\pi \frac{{(i - k)m}}{K}\alpha }}} } }_{fake} + {n_k}, k \in {\cal K},
\end{align}
where $ { x_i},i \in \mathcal{ K} $ denotes the true transmit signals and  $ {\tilde x_i},i \in \mathcal{\tilde K} $ are specified fake decoy signals, where $\mathcal{ K} $ and $\mathcal{\tilde K} $ represent the sets for true and fake frequency deployment, respectively, $\alpha$ denotes the T/F frequency multiplexing factor,
$ {h_i} \sim {\cal{CN}} \left( {{u_{{h_i}}}, \sigma _i^2} \right) $ is the channel response from the $i^{th}$ user to Bob with mean $ {{u_{{h_i}}}} $ and variance ${\delta _i^2}$, and $ {n_k} $ is the received nose from the $i^{th}$ frequency band with power $ \delta _k^2$.

Since the decoy signals are deliberately set for a specific purpose, we assume that Bob is aware of them and can eliminate the interference caused by these signals.
Therefore, the $ {y_k} $ at Bob can be rewritten as:
\begin{equation}\label{eq2}
{y_k} = \sqrt {{p_k}} {h_k}{x_k} + \sum\limits_{m = 0}^{K - 1} {\sum\limits_{i \ne k}^{ {\cal K}} {\frac{{\sqrt {{p_i}} }}{K}{h_i}{x_i}{e^{j2\pi \frac{{(i - k)m}}{K}\alpha }}} }  + {n_k},k \in {\cal K}.
\end{equation}

Based on (\ref{eq2}), the received signal-to-interference-plus-noise ratio (SINR) for detecting the true signal can be calculated as
\begin{equation}\label{eq3}
	{\gamma _k} = \frac{{{p_k}{{\left| {{h_k}} \right|}^2}}}{{\sum\limits_{i = 0,i \ne k,i \in \mathcal{ K}}^{K - 1} {{p_i}{{\left| {{h_i}} \right|}^2}{c_i}}  + \delta _k^2}}, k \in {\cal K},
\end{equation}
where $ c_i $ represents the correlation coefficient between signals in both the $ i^{th} $ and the $ k^{th} $ frequency band which can be calculated as
\begin{equation}\label{eqci}
	 {c_i} = {\left| {\frac{{{\rm{sinc}}\left( {\alpha \left( {i - k} \right)} \right)}}{{{\rm{sinc}}\left( {\alpha \left( {i - k} \right){\rm{ }}/K} \right)}}} \right|^2}.
\end{equation}

At the intercepting end, to achieve the purpose of deception, we expect the Eve to receive the decoy signal, represented as
\begin{equation}\label{eq4}
{y_{e,n}} = \sqrt {{p_n}} {h_{e,n}}{{\tilde x}_n} + \frac{1}{K}\sum\limits_{m = 0}^{K - 1} {\sum\limits_{i \ne n}^{{\cal K} \cup \tilde {\cal K}} {\sqrt {{p_i}} {h_{e,i}}{x_i}{e^{\frac{{j2\pi (i - n)m}}{K}\alpha }}} }  + {n_{e,n}}.
\end{equation}

Deriving from (\ref{eq4}), the intercepting SINR at Eve from the $ k^{th} $ frequency band can be represented as 
\begin{equation}\label{eq5}
{\gamma _{e,k}} = \frac{{{p_k}{{\left| {{h_{e,k}}} \right|}^2}}}{{\sum\limits_{i = 0,i \ne k}^{K - 1} {{p_i}{{\left| {{h_{e,i}}} \right|}^2}{c_i}}  + \delta _{e,k}^2}}, k \in {\cal K},
\end{equation}
and the SINR of deception at Eve from the $ n^{th} $ frequency band is also represented as 
\begin{equation}\label{inSINR}
{{{\tilde \gamma }_{e,n}}} = \frac{{{p_n}{{\left| {{h_{e,n}}} \right|}^2}}}{{\sum\limits_{j = 0,i \ne n}^{K - 1} {{p_j}{{\left| {{h_{e,j}}} \right|}^2}{c_i}}  + \delta _{e,n}^2}},n \in \tilde {\cal K}.
\end{equation}

\begin{remark}\label{R1}
	To achieve the goal of deception, it is necessary to ensure that the decoy signal serves as the primary demodulation signal at the non-cooperative receiver. This requires satisfying the condition ${{{\tilde \gamma }_{e,n}}} \ge {\gamma _{e,k}}, n \in \tilde {\cal K}$ and $k \in {\cal K}$. 
	Specifically, let \( f_i = p_i \left| h_{e,i} \right|^2 \) and \( \varphi = \sum_{i=0}^{K-1} p_i \left| h_{e,i} \right|^2 c_i + 1 \). Then we have
	$  {{{\tilde \gamma }_{e,n}}} = \frac{f_n}{\varphi - f_n}, \quad {\gamma _{e,k}} = \frac{f_k}{\varphi - f_k} $, 
	where \(\frac{x}{\varphi - x}\) is a monotonically increasing function. Consequently, we obtain
	 \begin{equation}\label{eq7}
	 	p_n \left| h_{e,n} \right|^2 \ge p_k \left| h_{e,k} \right|^2,
	 \end{equation}
which indicates the strength of decoy signal received by Eve should be greater than that of the true signal it attempts to intercept.
\end{remark}

Based on Remark 1, the interception probability can be defined by 
\begin{equation}\label{eq7}
%	\Pr \left( {{ \gamma _{e,k}} \ge {{\tilde T}_h},{p_n}{{\left| {{h_{e,n}}} \right|}^2} \le {p_k}{{\left| {{h_{e,k}}} \right|}^2}} \right),
	\Pr \left( { \gamma _{e,k}} \ge {{\tilde T}_h} \right),
\end{equation}
which means the probability that the SINR of true signal intercepted by Eve exceeds a certain threshold $ {{\tilde T}_h} $. 
$ {{\tilde T}_h} $ ensures the deceptive quality constraint. $ \Pr \left(  \cdot  \right) $ represents the probability calculation.
In this case, we assume that Eve can obtain information from intercepted true signals.
In the other case, the probability that Eve being deceived is defined as 
 \begin{equation}\label{eq8}
\Pr \left( {{{\tilde \gamma }_{e,n}} \ge {{\tilde T}_h},{p_n}{{\left| {{h_{e,n}}} \right|}^2} \ge {p_k}{{\left| {{h_{e,k}}} \right|}^2}} \right),
% 	\Pr \left( { \gamma _{e,k}} \ge {T_h} \right),
 \end{equation} 
which indicates that under the condition where the decoy signals arrived at Eve dominate, the SINR of the decoy signals received by Eve must exceed the certain threshold.

Comparing (4) and (6), we can make the following indications: multi-source cooperation can eliminate the interference caused at Bob by a prescient information of fake signals, resulting in the SINR being reduced only by the noise and co-channel interference from other true signals at different frequency bands. Whereas the SINR received at Eve will be affected by the interference from all signals transmitted from full spectrum.

To realize the anti-interception transmission with achieving the purpose of CTSF, we design a transmission that maximizes the secrecy rate to counteract the interception probability, thereby achieving the goal of concealing the truth. Meanwhile, we ensure the signal quality of the decoy signals received by Eve to enhance the deception probability, thereby achieving the goal of showing the fake. To this end, we formulate an optimization problem as follows
\begin{subequations}\label{eq9}
\begin{align}
		\mathcal{P}1: \max_{\{p_i\},\alpha}& \quad \sum\limits_{k \in \mathcal{ K}} {{{\log }_2}\left( {1 + {\gamma _k}} \right) - {{\log }_2}\left( {1 + {\gamma _{e,k}}} \right)}  \label{9a} \\
	\text{s.t.:}& \quad {{{\tilde \gamma }_{e,n}} \ge {{\tilde T}_h}},n \in \mathcal{\tilde K} \label{9b}, \\
	& \quad p_n \left| h_{e,n} \right|^2 \ge p_k \left| h_{e,k} \right|^2, n \in \mathcal{\tilde K}, k \in {\cal K}, \label{9c}\\
	&  \quad 0 \le \sum\limits_{i = 0}^{K - 1} {{p_i}}  \le {P_s}, \label{9d}
\end{align}
\end{subequations}
where (\ref{9a}) serves as the objective function to maximize the sum secrecy rate of true signals, employing an information-theoretic security metric that embodies the ``conceal truth" principle in our CTSF framework; the constraint (\ref{9b}) guarantees effective deception of Eve by maintaining the deception metric above a predefined threshold $ {\tilde T}_h $, corresponding to the ``show fake" requirement;
 (\ref{9c}) ensures the received decoy signals dominate at Eve’s side (based on (\ref{eq8}), it is a necessary condition for successful deception.); and the constraint (\ref{9d}) is to limit the maximum transmission power for transmitting such true and fake signals.

\section{Solution and Discussion}%and Problem Statement 
This section mainly presents the solution method designed for the aforementioned problem and provides some in-depth discussions. Specifically, we first perform a series of mathematical transformations, variable substitutions, and simplifications for the aforementioned non-convex optimization problem. Then, we propose a two-stage alternating iterative optimization method, and finally, we propose a single-variable nonlinear optimization method based on a Newton's method to fit the frequency reuse parameters. Additionally, we provide some theorems to discuss some indications.

\begin{figure*}
	\begin{align}\label{eqob}
		{R_s} 
		&= \left| {\cal K} \right|{\log _2}\big( {\sum\limits_{i = 0,i \in {\cal K}}^{K - 1} {{\xi _i}{{\left| {{h_i}} \right|}^2}}  + 1} \big) + \sum\limits_{k \in {\cal K}} {{{\log }_2}\big( {\sum\limits_{i = 0,i \ne k}^{K - 1} {{\xi _i}{{\left| {{h_{e,i}}} \right|}^2}}  + 1} \big)}  - \left| {\cal K} \right|{\log _2}\big( {\sum\limits_{i = 0}^{K - 1} {{\xi _i}{{\left| {{h_{e,i}}} \right|}^2}}  + 1} \big) \nonumber\\
		&- \sum\limits_{k \in {\cal K}} {{{\log }_2}\big( {\sum\limits_{i = 0,i \ne k,i \in {\cal K}}^{K - 1} {{\xi _i}{{\left| {{h_i}} \right|}^2}}  + 1} \big)} \nonumber\\
		&= \left| {\cal K} \right|{\log _2}\big( {\tau \sum\limits_{i = 0,i \in {\cal K}}^{K - 1} {{\xi _i}{{\left| {{h_i}} \right|}^2}}  + \tau } \big) + \sum\limits_{k \in {\cal K}} {{{\log }_2}\big( {{\mu _k}\sum\limits_{i = 0,i \ne k}^{K - 1} {{\xi _i}{{\left| {{h_{e,i}}} \right|}^2}}  + {\mu _k}} \big)}.
	\end{align}
	\hrulefill % 添加一条水平线
\end{figure*}

By using (\ref{eq3}), (\ref{eq5}) and (\ref{inSINR}), the objective function in (\ref{9a}) can be simplified in (\ref{eqob}) as shown at the top of this page.
In (\ref{eqob}), $ \left| {\cal K} \right| $ denotes the length of set $ {\cal K} $, and
the variables substitution are made by
\begin{equation}\label{RV1}
\tau  = {1 \mathord{\left/
		{\vphantom {1 {\left( {\sum\limits_{i = 0}^{K - 1} {{\xi _i}{{\left| {{h_{e,i}}} \right|}^2}}  + 1} \right)}}} \right.
		\kern-\nulldelimiterspace} {\bigg( {\sum\limits_{i = 0}^{K - 1} {{\xi _i}{{\left| {{h_{e,i}}} \right|}^2}}  + 1} \bigg)}},	
\end{equation}
and 
\begin{equation}\label{RV2}
{\mu _k} = {1 \mathord{\left/
		{\vphantom {1 {\left( {\sum\limits_{i \ne k}^{\cal K} {{\xi _i}{{\left| {{h_i}} \right|}^2}}  + 1} \right)}}} \right.
		\kern-\nulldelimiterspace} {\bigg( {\sum\limits_{i \ne k}^{\cal K} {{\xi _i}{{\left| {{h_i}} \right|}^2}}  + 1} \bigg)}}, k\in \cal{K}.
\end{equation}

% $ \tau  = 1/(\sum\limits_{i = 0}^{K - 1} {{\xi _i}{{\left| {{h_{e,i}}} \right|}^2}}  + 1) $, and 
% ${\mu _k} = 1/(\sum\limits_{i \ne k}^{\cal{K}} {{\xi _i}{{\left| {{h_i}} \right|}^2}}  + 1), k \in \cal{K}$.	

%\begin{equation}\label{key}
%{\bf{\xi }} = \left[ {{\xi _{0,...,}}{\xi _i},...} \right],{\bf{g}} = {\left[ {{{\left| {{h_0}} \right|}^2}_{,...,}{{\left| {{h_i}} \right|}^2},...} \right]^T},{{\bf{g}}_e} = {\left[ {{{\left| {{h_{e,0}}} \right|}^2}_{,...,}{{\left| {{h_{e,i}}} \right|}^2},...} \right]^T}
%\end{equation}
%\begin{equation}\label{eq12}
%{R_s} = \left| {\cal K} \right|{\log _2}\left( {\tau \sum\limits_{i = 0,i \in {\cal K}}^{K - 1} {{\xi _i}{{\left| {{h_i}} \right|}^2}}  + \tau } \right) + \sum\limits_{k \in {\cal K}} {{{\log }_2}\left( {{\mu _k}\sum\limits_{i = 0,i \ne k}^{K - 1} {{\xi _i}{{\left| {{h_{e,i}}} \right|}^2}}  + {\mu _k}} \right)} 
%\end{equation}
%with constraints
%\begin{equation}\label{c1}
%\tau \sum\limits_{i = 0}^{K - 1} {{\xi _i}{{\left| {{h_{e,i}}} \right|}^2}}  + \tau  - 1 \le 0,
%\end{equation}
%and 
%\begin{equation}\label{c2}
%{\mu _k}\sum\limits_{i = 0,i \ne k,i \in K}^{K - 1} {{\xi _i}{{\left| {{h_i}} \right|}^2}}  + {\mu _k} - 1 \le 0, k \in {\cal K}.
%\end{equation}

By replacing {\color{blue}${\xi _i} = {p_i}{c_i}$ $ \forall i$}, the problem $ \mathcal{P}1 $ can then be reformulated as
\begin{subequations}\label{eq16}
	\begin{align}
	\mathcal{P}2:& \quad \max_{{\xi_i}, \tau, \mu_k} {R_s} \label{16a} \\
	\text{s.t.:} &  \quad  {{{\tilde T}_h}}\sum\limits_{i \ne n}^{{\cal K} \cup \tilde {\cal K}} {{\xi _i}{{\left| {{h_{e,i}}} \right|}^2}}  + {{{\tilde T}_h}} \le {\xi _n}{\left| {{h_{e,n}}} \right|^2},n \in \tilde {\cal K}, \label{16b} \\
	& \tau \sum\limits_{i = 0}^{K - 1} {{\xi _i}{{\left| {{h_{e,i}}} \right|}^2}}  + \tau  - 1 \le 0, \label{16c}\\
	& {\mu _k}\sum\limits_{i = 0,i \ne k,i \in K}^{K - 1} {{\xi _i}{{\left| {{h_i}} \right|}^2}}  + {\mu _k} - 1 \le 0, k \in {\cal K}, \label{16d}\\
	& {\xi _n} \ge {\xi _k}, n \in \mathcal{\tilde K}, k \in {\cal K}, \label{16e} \\
	& 0 \le \sum\limits_{i = 0}^{K - 1} {{\xi _i}}  \le {P_s}. \label{16f}
	\end{align}
\end{subequations}

Due to the presence of the terms \(\xi_i\) and \(\tau, \mu_k\) being multiplied together in (\ref{16a}), (\ref{16c}), and (\ref{16d}), the problem P2 remains non-convex. However, it conforms to a mathematical structure of multi-variable bi-convex optimization. Therefore, we propose an alternating iterative framework to solve this problem. Particularly, a bi-stage convex optimization framework for \(\xi_i\) and $ (\tau, \mu_k) $ is conducted.  

In the first stage, given an initial $ \left\{ {\xi _i^ \circ } \right\} $, we find solutions of introduced variables, i.e., $ \tau$ and $\mu_k\ $. In this case, the problem $ \mathcal{P}2 $ is reformulated as
\begin{subequations}\label{eq10}
	\begin{align}
		\mathcal{T}1:& \quad \max_{\tau, \mu_k} \left| \mathcal{K} \right|\ln \left( {a\tau } \right) + \ln \big( {\sum\limits_{k \in \mathcal{K}} {\ln \left( {b{\mu _k}} \right)} } \big) \label{10} \\
		\text{s.t.:} &  \quad  (\ref{16c}),(\ref{16d}), \label{10a} 
	\end{align}
\end{subequations}
where $a = \sum\limits_{i \in \mathcal{K}} {\xi _i^ \circ {{\left| {{h_i}} \right|}^2}}  + 1$, $ b =\sum\limits_{i \ne k}^{\cal K} {\xi _i^\circ {{\left| {{h_{e,i}}} \right|}^2}}  + 1$, and $ \left\{ {\xi _i^ \circ } \right\} $ denotes a feasible solution of $ \left\{ {\xi _i } \right\} $ which can be initialized at the beginning of algorithm.

For the second stage, substituting the obtained $\tau ^ \circ$ and $\mu _k^ \circ $ by solving $\mathcal{T}1  $, the problem $ \mathcal{P}2 $ is then converted to 
\begin{subequations}\label{eq18}
	\begin{align}
		\mathcal{T}2:& \quad \max_{{\xi_i}} {R_s}\left( {{\xi _i},{\tau ^ \circ },\mu _k^ \circ } \right) \label{18a} \\
		\text{s.t.:} &  \quad  (\ref{16b}-\ref{16e}). \label{18b}
	\end{align}
\end{subequations}

After the above reformulations and simplifications, the two sub-problems $ \mathcal{T}1 $ and $ \mathcal{T}2 $ have been converted into convex problems, which can be solved directly. We use the MOSEK solver to solve them separately and perform alternating iterations until convergence. As shown in Algorithm 1, stage 1 is for solving $ \mathcal{T}1 $ and stage 2 is for solving $ \mathcal{T}2 $.
The specific algorithm execution process is as follows. Each iteration of the alternating method has a complexity of \( O(2K^3)\), where \(K\) is the number of sources. Let \(t\) be the number of iterations required for convergence. The total complexity is \(O(2t \cdot K^3)\).
\begin{algorithm}[h]
	\caption{Bi-Stage Alternating Dual-Domain Optimization (BADO) Approach}
	%	\caption{Joint power-frequency dual domain optimization: Bi-stage alternating iterative approach}
	\KwIn{Initial values \(\left\{ \xi_i^\circ \right\}\) for the variables \(\left\{ \xi_i \right\}\).}
	\KwOut{\(\left\{ \xi_i^\star \right\}\)--optimal values of \(\left\{ \xi_i \right\}\), \(\tau^\star\)--optimal value of \(\tau\), and 
			\(\left\{ \mu_k^\star \right\}\)--optimal value of \(\left\{\mu_k\right\}\).}
	\Repeat{convergence}{
		\textbf{Stage 1: Solve for \(\tau\) and \(\mu_k\)}\;
	Given the current \(\left\{ \xi_i^\circ \right\}\) (Initial value can be utilized at the procedure startup, with subsequent updates performed using the output from stage 2), solve the problem $ \mathcal{T}1 $ using a convex MOSEK solver\;
	Solved results: \(\tau^\circ\) and \(\mu_k^\circ\)\;
			Update \(\tau\) and \(\mu_k\) to  \(\tau^\circ\) and \( \left\{\mu_k^\circ\right\}\) in problem $ \mathcal{T}2 $\;		
			\textbf{Stage 2: Solve for \(\xi_i\)}\;
			Given the updated \(\tau^\circ\) and \(\mu_k^\circ\), solve the problem $ \mathcal{T}2 $ using a convex MOSEK solver\;
			Solved results: \(\left\{ \xi_i^\circ \right\}\)\;
			Update \(\left\{\xi_i\right\}\) to \(\left\{ \xi_i^\circ \right\}\) and return to Stage 1\;
		}
	\Return \(\left\{ \xi_i^\star \right\}\)  $ \gets $ \(\left\{ \xi_i \right\}\), \(\tau^\star\)  $ \gets $ \(\tau\), and \(\left\{ \mu_k^\star \right\}\) $  \gets $ \(\left\{ \mu_k \right\}\).\
	\end{algorithm}

Keep (\ref{eqci}) in mind, we further to fit a $ \alpha $ based on the value of $ {c_i} (i = 0,..., K-1) $. Newton's method is an iterative technique for solving nonlinear equations or optimization problems. It leverages the first and second derivatives of the objective function to quickly converge to the function's extremum. For single-parameter optimization problems, the core idea of Newton's method is to approximate the objective function using a Taylor series expansion and iteratively update the parameter to find the optimal solution. Particularly, consider the objective function we aim to optimize
\begin{equation}\label{eq16}
	f(\alpha) = \sum_{i=0}^{K-1} \big( c_i - \big( \frac{\text{sinc}(\alpha (i - k))}{\text{sinc}\left(\frac{\alpha (i - k)}{K}\right)} \big)^2 \big)^2,
\end{equation}
where $ \{c_i\} $ are known data, and \( K \) and \( k \) are constants. We seek the parameter \( \alpha \) that minimizes \( f(\alpha) \).
Its gradient is calculated as 
\begin{equation}\label{eq17}
f'(\alpha) = \frac{d}{d\alpha} \sum_{i=0}^{K-1} \big( c_i - \big( \frac{\text{sinc}(\alpha (i - k))}{\text{sinc}\big(\frac{\alpha (i - k)}{K}\big)} \big)^2 \big)^2.	
\end{equation}

Its Hessian is calculated as 
\begin{equation}\label{eq200}
	f''(\alpha) = \frac{d^2}{d\alpha^2} \sum_{i=0}^{K-1} \big( c_i - \big( \frac{\text{sinc}(\alpha (i - k))}{\text{sinc}\left(\frac{\alpha (i - k)}{K}\right)} \big)^2 \big)^2.
\end{equation}

Then an iterative update is executed by
\begin{equation}\label{eq211}
\alpha_{k+1} = \alpha_k - \frac{f'(\alpha_k)}{f''(\alpha_k)}.	
\end{equation}
If \( |f'(\alpha_k)| \) or \( |\alpha_{k+1} - \alpha_k| \) is less than the predefined threshold, stop the iteration and return the optimal parameter \( \alpha \).	
\begin{algorithm}[t]
	\caption{Newton's Method for Single Parameter Optimization}
	\label{alg:newton_method}
	\KwIn{Initial guess $\alpha_0$, tolerance $\epsilon_1$, $\epsilon_2$, maximum iteration number $N$}
	\KwOut{Optimal parameter $\alpha^*$}
	$k \gets 0$\;
	\While{$k < N$}{
		Compute gradient $g_k \gets f'(\alpha_k)$ by using (\ref{eq17})\;
		Compute Hessian $H_k \gets f''(\alpha_k)$ by using (\ref{eq200})\;
		Update parameter $\alpha_{k+1} \gets \alpha_k - \frac{g_k}{H_k}$ by using (\ref{eq211})\;

		\If{$|g_k| < \epsilon_1$ \textbf{or} $|\alpha_{k+1} - \alpha_k| < \epsilon_2$}{
			\textbf{break}\;
		}
		$k \gets k + 1$\;
	}
	$\alpha^* \gets \alpha_k$\;
	\Return $\alpha^*$\;
\end{algorithm}

Let the denominators of \(\gamma_k\) and \(\gamma_{k,e}\) be constants \( C_k \) and \( C_{k,e} \), respectively,
\begin{equation}\label{eq18}
	C_k = \sum\limits_{i = 0, i \ne k, i \in \mathcal{K}}^{K-1} p_i \left| h_i \right|^2 c_i + 1,
\end{equation}
\begin{equation}\label{eq19}
	C_{k,e} = \sum\limits_{i = 0, i \ne k}^{K-1} p_i \left| h_{e,i} \right|^2 c_i + 1.
\end{equation}

Based on the results optimized by Algorithm 1, the values of $ C_k $ and $ C_{k,e} $ in (\ref{eq18}) and (\ref{eq19}), respectively, can be determined. The power allocation for true signals $ {p_k}, k \in \cal{K} $ can be calculated by the following equation set.
\begin{equation}\label{eq28}
	{R_{s,k}} = {\log _2}\left( {\frac{{{C_{k,e}}\left( {{C_k} + {p_k}{\alpha _k}} \right)}}{{{C_k}\left( {{C_{k,e}} + {p_k}{\alpha _{k,e}}} \right)}}} \right), k \in \cal{K}.
\end{equation}
%Besides, based on the Theorem 1 and Theorem 2, we can find that the maximum objective value is reached when $ {\gamma _{n,e}} \ge  {{\tilde T}_h},n \in \mathcal{\tilde K} $. 
The power allocation for fake signals $ {p_n}$ can be calculated by $ {p_n} =  {{\tilde T}_h}{{\left( {{C_{n,e}} + 1} \right)} \mathord{\left/
		{\vphantom {{\left( {{C_{n,e}} + 1} \right)} {{\alpha _{n,e}}}}} \right.
		\kern-\nulldelimiterspace} {{\alpha _{n,e}}}},n \in \tilde{\cal{K}}$. Thus, $ {c_i} = {{{\xi _i}} \mathord{\left/
		{\vphantom {{{\xi _i}} {{p_i}}}} \right.
		\kern-\nulldelimiterspace} {{p_i}}} (i = 0,..., K-1)$ can be obtained.

Checking the convergence involves comparing the absolute value of the gradient \( |f'(\alpha_k)| \) or the step size \( |\alpha_{k+1} - \alpha_k| \) with a threshold. This is a constant-time operation, \( O(1) \).
Let \( N \) be the number of iterations required for convergence. Each iteration involves computing the gradient and Hessian, both of which have a complexity of \( O(K) \).
The parameter update and convergence check are both \( O(1) \).
Therefore, the complexity per iteration is \( O(K) \), and the total complexity of the algorithm is
$ O(NK) $.

Different from common research work on physical layer security, this work not only ensures the secure transmission of the legitimate signal but also aims to achieve the deception of the false signal. Based on the original problem modeling, we have the following analytical results.
\begin{theorem}
As the transmission power \( p_k \) of the \( k \)-th user increases, the secrecy rate \( R_{s,k} \) of the \( k \)-th user increases monotonically, but the rate of increase gradually decreases.
\end{theorem}
\begin{proof}
Based on (\ref{eq4}), \(\gamma_k\) increases with \( p_k \) because the numerator \( p_k \left| h_k \right|^2 \) increases while the denominator remains constant, considering the other power allocations are determined.
 \(\gamma_{e,k}\) in (\ref{eq5}) also increases with \( p_k \).

Note that properties of the logarithmic function, \(\log_2 (1 + x)\) is a monotonically increasing function of \( x \), but its rate of increase decreases as \( x \) increases.
%4. \textbf{Effect on the secrecy rate \( R_{s,k} \) of the \( k \)-th user}:
Since \(\gamma_k\) and \(\gamma_{k,e}\) increase with \( p_k \) and the logarithmic function is monotonically increasing, \(\log_2(1 + \gamma_k)\) and \(\log_2(1 + \gamma_{k,e})\) also increase with \( P_k \). However, the rate of increase of \(\log_2(1 + \gamma_k)\) decreases as \(\gamma_k\) increases, and similarly, the rate of increase of \(\log_2(1 + \gamma_{k,e})\) decreases as \(\gamma_{k,e}\) increases.

Let \( \gamma_k = \frac{a_k p_k}{C_k} \) and \( \gamma_{k,e} = \frac{a_{k,e} p_k}{C_{k,e}} \), where \( a_k = \left| h_k \right|^2 \) and \( a_{k,e} = \left| h_{e,k} \right|^2 \).
Calculate the first derivative of \( R_{s,k} \) with respect to \( p_k \):
\begin{align}\label{eq20}
\frac{dR_{s,k}}{dp_k} &= \frac{d}{dp_k} \left[ \log_2 \left(1 + \frac{a_k p_k}{C_k}\right) - \log_2 \left(1 + \frac{a_{k,e} p_k}{C_{k,e}}\right) \right]\nonumber	\\
&= \frac{a_k}{C_k \ln(2) \left(1 + \frac{a_k p_k}{C_k}\right)} - \frac{a_{k,e}}{C_{k,e} \ln(2) \left(1 + \frac{a_{k,e} p_k}{C_{k,e}}\right)}.
\end{align}
 We can see that \(\frac{dR_{s,k}}{dp_k} > 0\), i.e., \( R_{s,k} \) increases monotonically with \( p_k \). 
We further calculate the second derivative of \( R_{s,k} \) with respect to \( p_k \):
\begin{align}\label{key}
\frac{{{d^2}{R_{s,k}}}}{{dp_k^2}} &= \frac{1}{{\ln (2)}}\frac{d}{{d{p_k}}}\left[ {\frac{{{a_k}}}{{{C_k}\left( {1 + \frac{{{a_k}{p_k}}}{{{C_k}}}} \right)}} - \frac{{{a_{k,e}}}}{{{C_{k,e}}\left( {1 + \frac{{{a_{k,e}}{p_k}}}{{{C_{k,e}}}}} \right)}}} \right] \nonumber\\
&=  - \frac{1}{{\ln (2)}}\big( {\frac{{a_k^2}}{{C_k^2{{\left( {1 + \frac{{{a_k}{p_k}}}{{{C_k}}}} \right)}^2}}} + \frac{{a_{k,e}^2}}{{C_{k,e}^2{{\left( {1 + \frac{{{a_{k,e}}{p_k}}}{{{C_{k,e}}}}} \right)}^2}}}} \big).
\end{align}
The second derivative is negative, indicating that the rate of increase of \( R_{s,k} \) gradually decreases as \( p_k \) increases.
Therefore, as the transmission power \( p_k \) of the \( k \)-th user increases, the secrecy rate \( R_{s,k} \) of the \( k \)-th user increases monotonically, but the rate of increase gradually decreases.

\end{proof}

\begin{theorem}
	As the threshold ${{\tilde T}_h}$ for SINR at the intercepting end increases, under the power budget constraint, the power of the decoy signal $p_n$, $n \in \tilde{\mathcal{K}}$ must increase to ensure the SINR requirement in (\ref{9b}) is met. This may also necessitate the reallocation of the power of the true signal $p_k$ to maintain the total power budget.
\end{theorem}

\begin{proof}
	For a total power budget \( P_{\text{total}} \), satisfying \( \sum_{i} p_i \le P_{\text{total}} \), when the threshold \( {{\tilde T}_h} \) increases, based on the constraint \( \gamma_{n,e} \ge {{\tilde T}_h} \), we need to recalculate the power of the decoy signal \( p_n \). According to Theorem 1, the power \( p_n \) must also increase as \( {{\tilde T}_h} \) increases when other power allocations are fixed. However, under the total power budget constraint \( P_{\text{total}} \), increasing \( p_n \) will affect the power allocation of other signals.
	If \( p_n \) increases, to satisfy \( \sum_{i} p_i \le P_{\text{total}} \), the power allocated to the true signals \( p_k \) must be reduced. This reallocation can be expressed as
	\begin{equation}
		P_{\text{total}} - \sum_{i \ne n} p_i \ge p_n.
	\end{equation}

It is also assumed that the SINR at the legitimate receiver, \(\gamma_k\), can completely eliminate the interference from the decoy signal, but in practice, there is residual interference. Therefore, the SINR for the true signal should be adjusted to
\begin{equation}\label{eq20}
	{\gamma _k} = \frac{{{p_k}{{\left| {{h_k}} \right|}^2}}}{{\sum\limits_{i \in {\cal K}} {{p_i}{{\left| {{h_i}} \right|}^2}{c_i}}  + \sum\limits_{n \in \tilde {\cal K}} {{{\tilde p}_n}{{\left| {{h_n}} \right|}^2}{c_n}}  + \delta _k^2}},k \in {\cal K}.
\end{equation}

%Therefore, an increase in the threshold \( Q \) for the intercepting end's SINR results in an increase in the power \( p_n \) of the decoy signal, which in turn necessitates a reduction in the power \( p_k \) of the true signals, leading to a decrease in their SINR and a consequent reduction in the secrecy rate.
The constraint \( \gamma_{n,e} \ge {{\tilde T}_h} \) implies that the SINR of the decoy signal at the intercepting end must reach or exceed the threshold \( {{\tilde T}_h} \). As \( {{\tilde T}_h} \) increases, the required SINR for the decoy signal also increases, necessitating an increase in the decoy signal power \( p_n \). Consequently, with a higher \({{\tilde T}_h} \), \( p_n \) must increase to ensure adequate SINR at the intercepting end. This increase in \( p_n \) impacts the overall power budget, potentially constraining the power allocation for the true signals \( p_k \). To maintain the total power budget, \( p_k \) may need to be reduced, depending on the power allocation strategy in optimization problem \( \mathcal{P}1 \). Efficient power allocation within the total power budget becomes crucial to meet all constraints, including \( \gamma_{n,e} \ge {{\tilde T}_h} \) and \( 0 \le p_i \le P_s \), without significantly reducing the SINR of the true signals.
\end{proof}

In addition, since $ {\xi _i} = {p_i}{c_i} $ in (\ref{eqob}), we further discuss the impact of the frequency multiplexing factor on the performance of both Bob and Eve. The greater the T/F frequency reuse among transmitted signals, the higher the correlation between the signals, i.e., \(c_i\) increases. The derivative of Bob's SINR with respect to \( c_i \) is
\begin{equation}
	\frac{\partial \gamma_k}{\partial c_i} = -\frac{p_k \left| h_k \right|^2 \cdot p_i \left| h_i \right|^2}{\left(\sum_{j \ne k, j \in \mathcal{K}} p_j \left| h_j \right|^2 c_j + \delta_k^2\right)^2} \le 0.
\end{equation}
The derivative of Eve's SINR with respect to \( c_i \) is
\begin{equation}
	\frac{\partial \gamma_{e,k}}{\partial c_i} = -\frac{p_k \left| h_{e,k} \right|^2 \cdot p_i \left| h_{e,i} \right|^2}{\left(\sum_{j \ne k} p_j \left| h_{e,j} \right|^2 c_j + \delta_{e,k}^2\right)^2} \le 0.
\end{equation}

For fixed transmission powers among all sources, \( \xi_i \) increases as \( c_i \) increases.
Thus, the SINR of the received true signal in (\ref{eq3}) decreases as \( c_i \), (\( i \ne k, i \in \mathcal{K} \)) increases, and the SINR in (\ref{eq5}) and (\ref{inSINR}) decreases as well.
To analyze the effect of the correlation factor \( c_i \) on the secrecy rate \( R_s \), we calculate the derivatives of the SINR for both Bob and Eve with respect to \( c_i \).
\begin{align}
	\frac{\partial R_s}{\partial c_i} &= \frac{\partial \left( \log_2(1 + \gamma_k) - \log_2(1 + \gamma_{k,e}) \right)}{\partial c_i} \nonumber\\
	&= \frac{1}{\ln(2)} \big( \frac{\partial \gamma_k / \partial c_i}{1 + \gamma_k} - \frac{\partial \gamma_{k,e} / \partial c_i}{1 + \gamma_{k,e}} \big) \nonumber\\
	&= \frac{{{p_k}{{\left| {{h_{e,k}}} \right|}^2} \cdot {p_i}{{\left| {{h_{e,i}}} \right|}^2}}}{{\ln (2)(1 + {\gamma _{e,k}}){{\big( {\sum\limits_{j \ne k} {{p_j}} {{\left| {{h_{e,j}}} \right|}^2}{c_j} + 1} \big)}^2}}} \nonumber\\
	&- \frac{{{p_k}{{\left| {{h_k}} \right|}^2} \cdot {p_i}{{\left| {{h_i}} \right|}^2}}}{{\ln (2)(1 + {\gamma _k}){{\big( {\sum\limits_{j \ne k,j \in {\cal K}} {{p_j}} {{\left| {{h_j}} \right|}^2}{c_j} + 1} \big)}^2}}}.
\end{align}

Analyzing the derivative \(\frac{\partial R_s}{\partial c_i}\), we observe the following:
\begin{itemize}
	\item As \( c_i \) increases, the interference for both Bob and Eve increases.
	\item If Bob can mitigate the interference more effectively than Eve, then the reduction in \(\gamma_k\) will be less significant than the reduction in \(\gamma_{k,e}\), leading to an increase in the secrecy rate \( R_s \).
\end{itemize}

\section{Numerical Results}%and Problem Statement 
We conduct simulations to demonstrate the effectiveness of our proposed scheme in an-intercepting transmissions. The simulation parameters were set as follows: the system involves four sources, denoted as $ K=4 $. Sources one and three operate at the true frequency and transmit genuine signals, while sources two and four operate at a false frequency and transmit decoy signals. 
The Rician fading channel model is employed for both legitimate and intercepting channels, with a Rician factor of 10 dB to account for strong line-of-sight components. All noise parameters are normalized with respect to the average received signal power, ensuring consistent SINR evaluation across different channel realizations.
	%The Rician channel model is adopted and the Rician factor is set to 10 dB, including both legitimate and eavesdropping channels. Normalized noise characterization is configured in our simulations.} 
	Since our proposed anti-interception scheme, which utilizes T/F frequency multiplexing, is a non-orthogonal frequency division multiplexing scheme, we have chosen OFDM as one of our benchmark systems. Additionally, an equal power allocation scheme has also been used as a benchmark, where the power of the true signal is equal to that of the decoy signal.

 \begin{figure}[ht]
	\centering
	\includegraphics[width=0.48\textwidth]{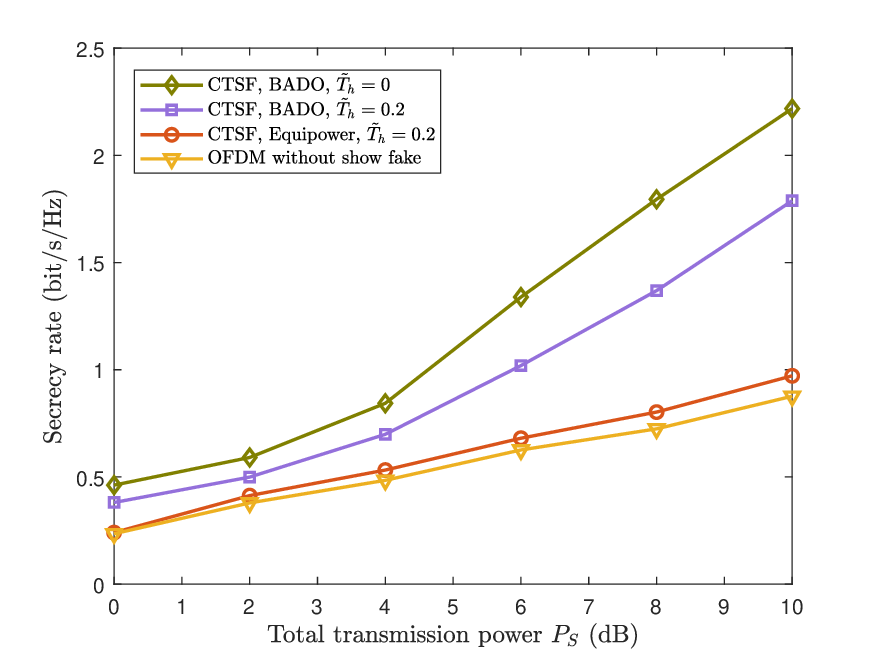}
	\caption {{The impact of total transmission power on secrecy rate performance.}}
	\label{fig2}
\end{figure}
Fig. \ref{fig2} shows the impact of total transmission power on the secrecy rate performance of true signals. From the Fig. \ref{fig2}, it is evident that the sum secrecy rate of true signals increases with the total transmission power. According to Theorem 1, as the total transmission power increases, more power can be allocated to true signals to maximize the sum secrecy rate, thereby enhancing the secrecy rate performance.
Our proposed CTSF scheme combined with the BADO approach outperforms the benchmark schemes, verifying its efficiency in both secrecy and deception performance. Compared to the OFDM scheme, the proposed CTSF with non-orthogonal interference between true and fake signals degrades the signal quality in the eavesdropping channel. Whereas, the equal power allocation approach, where the power of the true and decoy signals is the same, cannot always ensure the fulfillment of the constraint, i.e., ${{{\tilde \gamma }_{e,n}}} \ge {\gamma _{e,k}}, n \in \tilde {\cal K}$. Furthermore, as the total power increases, the power of the deceptive signal may be wasted under the given deception quality constraint, which subsequently limits the secure transmission performance of the true signal.

 \begin{figure}[ht]
	\centering
	\includegraphics[width=0.48\textwidth]{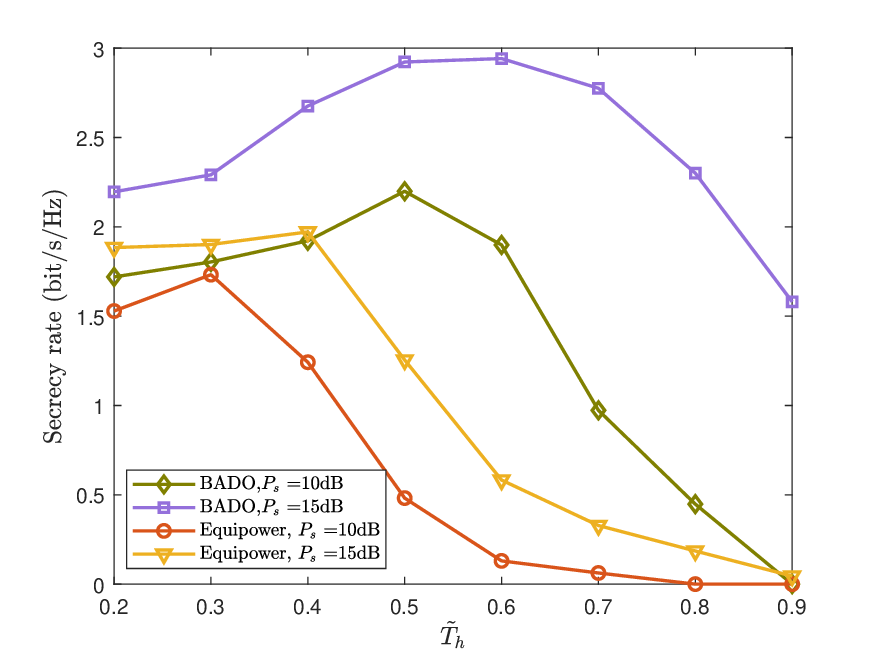}
	\caption {{Sum secrecy rate Vs. $ {{\tilde T}_h} $. }}
	\label{fig3}
\end{figure}
Fig. \ref{fig3} shows the impact of SINR threshold $ {{\tilde T}_h} $ constraint for decoy signals on secrecy rate performance. As illustrated in this figure, with the increase of 
$ {{\tilde T}_h} $, the sum secrecy rate of the true signals initially rises and then declines. Compared to the equal power allocation scheme, the proposed CSTF with BADO demonstrates superior secrecy rate performance. According to (7), as 
$ {{\tilde T}_h} $ increases, a higher power allocation to the deceptive signals is required to achieve the intended deception. This results in a reduction of the power allocated to the legitimate signals. However, when 
$ {{\tilde T}_h} $ is within a lower range, the interference caused by frequency reuse becomes the dominant factor. This interference significantly impairs the eavesdropping channel’s interception capability, which outweighs the reduction in the transmission power of the true signals. Consequently, the sum secrecy rate exhibits an upward trend. On the other hand, when 
$ {{\tilde T}_h} $ is within a higher range, more power is allocated to the deceptive signals to maintain deception quality, leaving less power for the transmission of legitimate signals, thereby causing a substantial decline in the secrecy rate.

 \begin{figure}[ht]
	\centering
	\includegraphics[width=0.48\textwidth]{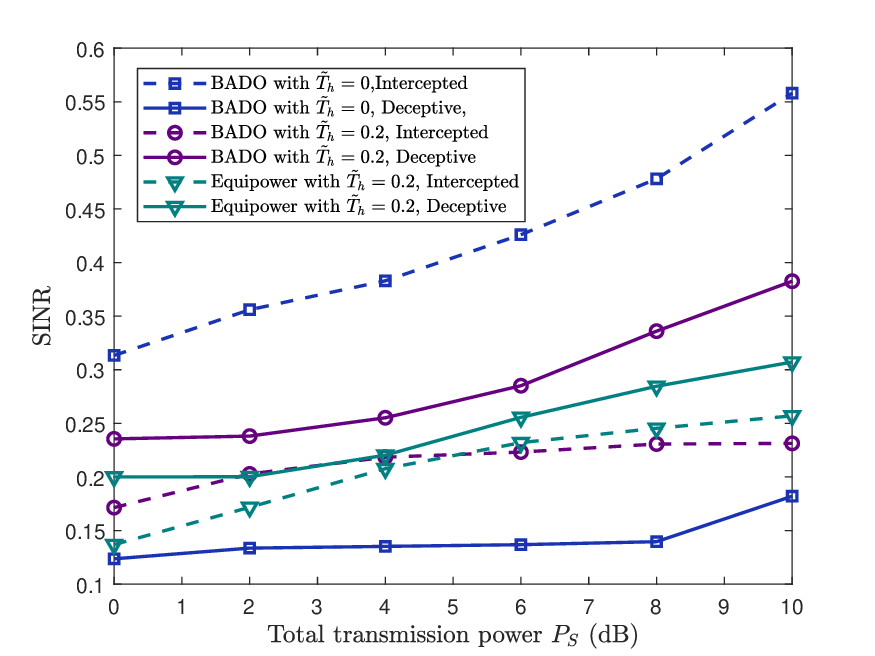}
	\caption {{The impact of total transmission power on the SINR of Eve.}}
	\label{fig4}
\end{figure}

Fig. \ref{fig4} shows the impact of total transmission power on Eve's SINR. We evaluate both the SINR of the true signal intercepted by Eve and the SINR of the decoy signal to which Eve is misled. The curves in Fig. \ref{fig4} reveal that when $ {{\tilde T}_h} =0$, indicating no decoy constraints, the intercepted SINR is significantly higher than the deceptive SINR. This observation, in conjunction with the results shown in Fig. \ref{fig2}, suggests that although a higher secure rate can be achieved under these conditions, the decoy objective is not met. When effective decoy constraints are applied, the decoy signal dominates Eve's received signal, resulting in a deceptive SINR that exceeds the intercepted SINR. This outcome aligns with the modeling of the proposed CTSF scheme. Furthermore, as the total transmission power increases, the SINR curves exhibit an upward trend, corroborating the analysis presented in Theorem 1. Comparatively, our proposed CTSF scheme demonstrates superior decoy effectiveness over equal power allocation strategies.

 \begin{figure}[ht]
	\centering
	\includegraphics[width=0.48\textwidth]{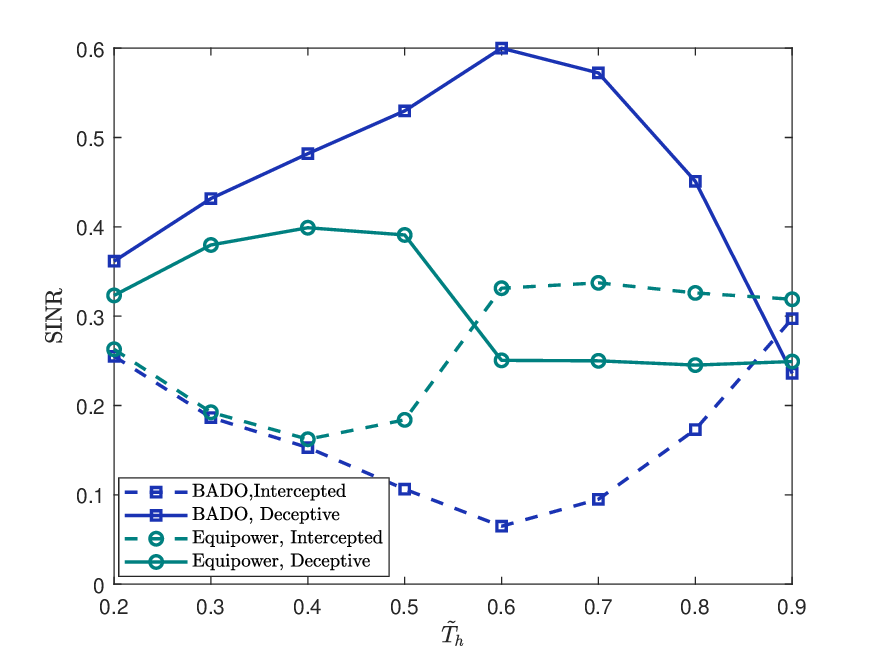}
	\caption {{SINR Vs. $ {{\tilde T}_h} $. ( $P_s = 10$ dB, $K = 4$)}}
	\label{fig5}
\end{figure}

Fig. \ref{fig5} shows the impact of the deceptive quality constraint threshold $ {{\tilde T}_h} $ on Eve's SINR. Similarly, we evaluate both the average SINR of the intercepted signal and the decoy signal. From Fig. \ref{fig5}, it can be observed that as $ {{\tilde T}_h} $ increases, the average deceptive SINR by the decoy signal initially increases and then decreases, while the average SINR of the intercepted signal initially decreases and then increases. This is because more power should be allocated to transmit the decoy signal to ensure its purpose and effectiveness, which in turn reduces the transmission power allocated to the true signal, leading to a decrease in the intercepted SINR.
 Notable intersections occur between interception and deception probability curves at ${{\tilde T}_h} \approx 0.56$ (equipower case) and ${{\tilde T}_h} \approx 0.88$ (BADO-based case). This indicates that the reception quality of true signals in the cooperative link achieves parity with the deception effectiveness of decoy signals. Notably, the intersection point generated by our proposed BADO method occurs at a higher ${{\tilde T}_h}$ value, demonstrating enhanced anti-interception capability that can counter sophisticated interceptors resistant to conventional deception tactics.

In addition to power allocation, the T/F frequency multiplexing factor determines the extent of spectrum overlap between the true and fake signals, affecting the level of interference between them and consequently impacting the SINR. Specifically, when $ {{\tilde T}_h} $ increases within a certain range, the power required for transmitting the decoy signal is relatively easy to satisfy, and according to the decoy condition in (\ref{9b}), i.e., $ {{{\tilde \gamma }_{e,n}} \ge {{\tilde T}_h}},n \in \mathcal{\tilde K} $, the deceptive SINR will increase with $ {{\tilde T}_h} $. However, when $ {{\tilde T}_h} $ continues to increase, not only must more power be tilted towards the decoy signal, but the interference between the true and decoy signals must also be reduced. Despite these adjustments, it becomes challenging to satisfy constraint (\ref{9b}), resulting in an increase in the intercepted SINR.

 \begin{figure}[ht]
	\centering
	\includegraphics[width=0.48\textwidth]{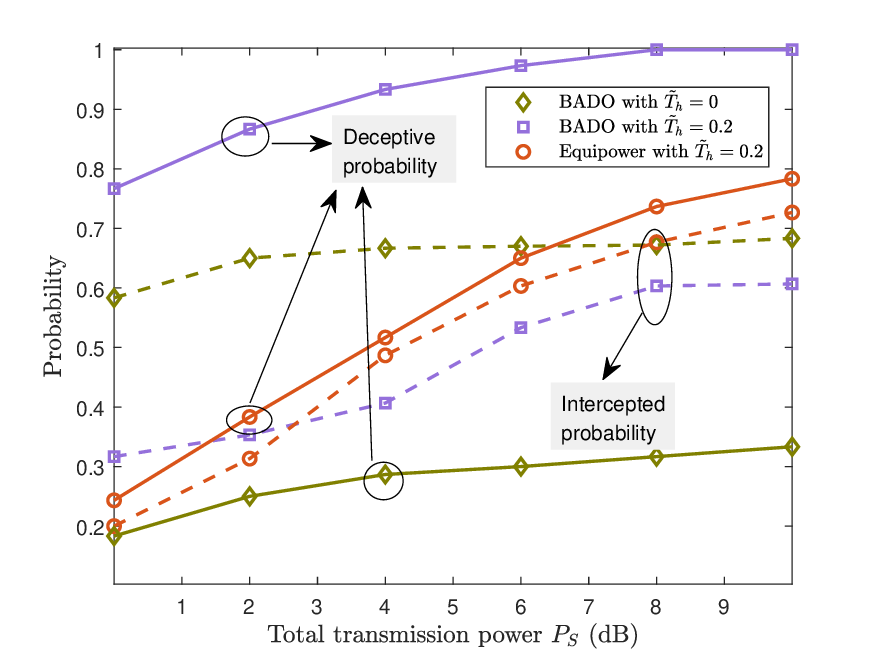}
	\caption {{The impact of total transmission power on the intercepted and deceptive probabilities. }}
	\label{fig6}
\end{figure}
In Fig. \ref{fig6}, we evaluate the impact of total transmission power on the intercepted and deceptive probabilities of the proposed CTSF scheme. It can be observed that as the total power increases, both the intercepted probability and the deceptive probability increase. This is because when more power is allocated to transmitting the false decoy signal, the decoy quality constraint threshold 
$ {{\tilde T}_h} $ is more easily satisfied. According to the definition in (\ref{eq8}), for a given 
$ {{\tilde T}_h} $, the decoy probability evidently increases with increasing power.
Although our objective function models the maximization of the secure transmission rate of the authentic signal, as the transmission power increases, the numerator in (\ref{eq5}) also increases, which in turn raises the probability of the true signal being intercepted.
In comparison to the equal power allocation scheme, our proposed BADO approach demonstrates a higher deceptive probability and a lower intercepted probability, indicating superior performance in CTSF. Additionally, without constraint on showing fake, the deceptive probability would be too low and the intercepted probability too high, failing to achieve the objective of CTSF.

 \begin{figure}[ht]
	\centering
	\includegraphics[width=0.48\textwidth]{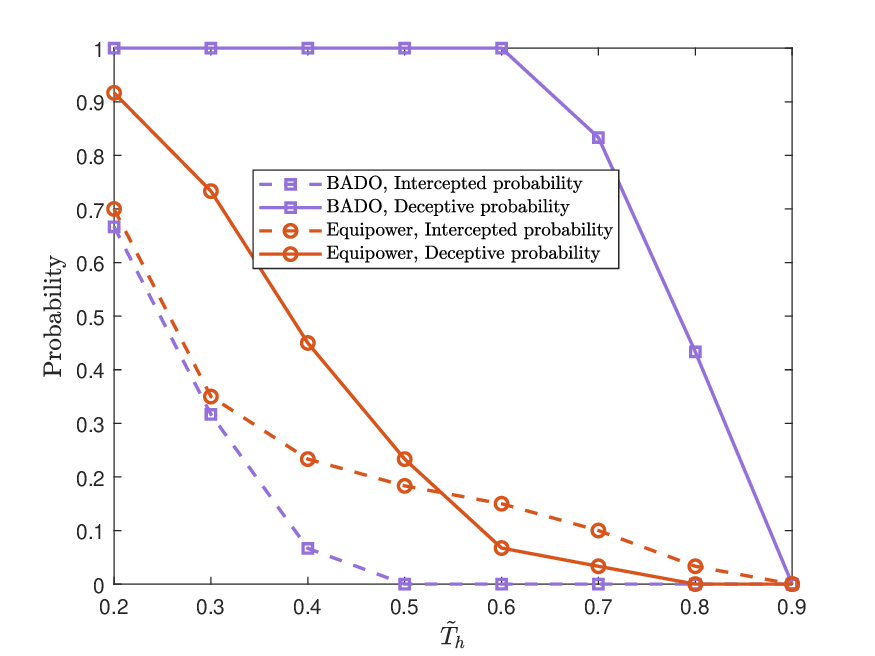}
	\caption {{The impact of $ {{\tilde T}_h} $ on the intercepted and deceptive probabilities.}}
	\label{fig7}
\end{figure}

Fig. \ref{fig7} shows the impact of the decoy condition constraint threshold $ {{\tilde T}_h} $ on the intercepted and deceptive probabilities. As $ {{\tilde T}_h} $ increases, the intercepted probability decreases rapidly. This is because ensuring the decoy condition requires more power allocation and reduced interference, which is consistent with the SINR value changes analyzed in Fig. \ref{fig5}. Compared to the equal power allocation scheme, the BADO method is validated to achieve a higher decoy probability and a lower interception probability. Specifically, when $ {{\tilde T}_h} $ is within a smaller range ($ {{\tilde T}_h} \le 0.6$), the decoy probability remains at 1, indicating that the decoy condition constraint can be fully satisfied within this range. However, as $ {{\tilde T}_h} $ continues to increase, it becomes challenging to meet the decoy condition constraint, resulting in a decrease in the decoy probability. 
In addition, Fig. \ref{fig7} demonstrates that when the deception threshold $ {{\tilde T}_h} $ lies within about [0.5, 0.6], the CTSF system achieves perfect deception performance with zero interception probability, realizing the optimal CTSF trade-off. The threshold $ {{\tilde T}_h} $ characterizes Eve's detection capability, where increased  values correspond to more sophisticated eavesdroppers requiring stronger deception signals.

\section{Conclusion}
To address interception threats in wireless adversarial communications, the CTSF for anti-interception transmission scheme based on T/F frequency multiplexing has been proposed in this paper. 
By multi-source cooperation, true signals carrying confidential information and decoy signals are transmitted on different frequencies, creating an overlap in the frequency domain. This approach can conceal true information while ensuring decoy signals mislead Eve.
Particularly, a problem has been formulated for maximizing the secrecy rate of true signals through joint optimization of user power allocation and T/F frequency multiplexing factor. Constraints ensure that decoy signals received by Eve are stronger than true signals, with a specified deceptive SINR threshold. Both BADO and Newton's methods are proposed to solve this optimization problem.
Finally, performance evaluations in terms of secrecy rate, intercepted probability, and deceptive probability demonstrate the effectiveness of our proposed CTSF scheme. 
 In future work, we will explore intelligent adversarial game-theoretic strategies within this integrated offensive-defensive anti-interception architecture, with adaptive countermeasures for diverse operational scenarios and mission-specific threat profiles.

\bibliographystyle{IEEEtran}
\bibliography{IEEEabrv,TF1}
\end{document}